\documentclass[utf8,latin1]{article}
\usepackage[utf8]{inputenc}
\usepackage{amsfonts,amsmath,amsthm,amssymb,bm,wasysym,comment}

\newtheorem{teor}{Theorem}
\newtheorem{remark}[teor]{Remark}
\newtheorem{lemma}[teor]{Lemma}
\newtheorem{cor}[teor]{Corollary}
\newtheorem {prop}[teor]{Proposition}

\newtheorem{definition}[teor]{Definition}
\newtheorem{example}[teor]{Example}
\newtheorem{notation}[teor]{Notation}

\newtheorem{conjecture}[teor]{Conjecture}

\usepackage[autostyle=true, italian=guillemets] {csquotes}
\usepackage[english]{babel}

\usepackage{graphics,standalone,tikz,pgfplots,textcomp,tikz}
\usepackage{geometry}
\pgfplotsset{compat=1.18} 

\DeclareMathOperator{\reg}{reg}

\DeclareMathOperator{\topp}{\mathrm{top}}
\DeclareMathOperator{\solvdeg}{\mathrm{solv.deg}}
\DeclareMathOperator{\LPP}{LPP}

\DeclareMathOperator{\lex}{lex}

\newcommand{\FF}{\mathbb{F}}

\usepackage{hyperref} 

\hypersetup{pdfpagemode=UseNone,colorlinks=true,linkcolor=blue,citecolor=blue,urlcolor=blue,allbordercolors=white,pdftitle={Reg LPP-ideals},pdfauthor={}}

\definecolor{bluegray}{rgb}{0.4, 0.6, 0.8}

\title{The complexity of solving a system of equations of the same degree\thanks{This research was funded by armasuisse through grant no. CYD-C-2020010.}}

\author{Giulia Gaggero and Elisa Gorla}  

\date{}

\begin{document}

\maketitle

\begin{abstract}
    Many systems of interest in cryptography consist of equations of the same degree. Under the assumption that the degree of regularity is finite, we prove upper bounds on the degree of regularity of a system of equations of the same degree, with or without adding the field equations to the system. The bounds translate into upper bounds on the solving degree of the systems, and hence on the complexity of solving them via Gr\"obner bases methods. Our bounds depend on the number of equations in the system, the number of variables, and the degree of the equations.
\end{abstract}

\section{Introduction}\label{intro}

Polynomial systems play an important role within cryptography, as several cryptographic schemes rely on the difficulty of solving multivariate polynomial equations over finite fields, especially in the post‑quantum setting. Most notably, multivariate cryptography relies on the difficulty of solving multivariate polynomial equations over finite fields as its core security assumption. The public key of a multivariate system is a polynomial system, which is easy to evaluate, but is believed to be infeasible to solve without secret structural information, giving rise to candidate quantum‑resistant encryption and signature schemes. Beyond multivariate schemes, computational problems can often be modeled via polynomial systems, whose solutions provide a solution to the original problem. MinRank~\cite{BFS96} is an example of such a problem, and the most efficient way of solving it currently consists of modeling it via a polynomial system~\cite{BBCGPSTV20,BBCPSV22,BBBGNRT20}. In algebraic cryptanalysis, block ciphers, hash functions, and other primitives are sometimes translated into large systems of polynomial equations, whose solutions allow an attacker to recover keys or find collisions, see e.g.~\cite{KS99, SKPI07, AC09, CW09, CP02, KLR24, CCMMP25}. 

The security of many cryptographic schemes, therefore, relies on the complexity of solving the related polynomial system. The degree of regularity of a polynomial system provides an upper bound on the complexity of solving the system using Gr\"obner bases~\cite{Sal23}, and hence an upper bound on the security of the corresponding scheme. However, often the degree of regularity is hard to compute or even estimate. In practice, one cannot compute it directly for systems of cryptographic size, as this would require a Gr\"obner basis computation for values of the parameters for which it is considered infeasible. Depending on the system in question, the degree of regularity and other related invariants may also be hard to estimate theoretically. Due to these difficulties, cryptographers often rely on heuristic estimates. A widespread approach consists of assuming that a system is semiregular (see e.g.~\cite{CFMPPR17}), which provides an explicit formula for its degree of regularity~\cite{BFS04}.

In this paper, we consider systems for which the degree of regularity is finite. These are exactly the systems for which the approach via the degree of regularity provides a non-trivial complexity bound. For polynomials in $n$ variables, the degree of regularity is finite if and only if the ideal generated by the top-degree part of the system contains a regular sequence of $n$ polynomials. For a system of polynomials of the same degree $D$, the degree of regularity is finite if and only if the regular sequence may be found in degree $D$. For such systems, in Theorem~\ref{teor:DD:LPP} we prove a sharp upper bound for the degree of regularity in terms of the number of equations, the number of variables, and the degree of the equations. An equivalent formula for the case $D=2$ appears in~\cite[Theorem 5.7]{BDDGMT21}. In Theorem~\ref{teor:D>q} we prove a similar result for a system of equations to which we add the field equations.

Our bounds on the degree of regularity readily produce bounds on the complexity of computing a Gr\"obner basis, hence of solving, many systems that are of interest in cryptography. Bounds on the complexity produced in this way have the advantage of being widely applicable and the disadvantage of not always being close to the actual complexity of solving the systems to which they apply. In particular, the complexity estimates they produce are less sharp than those produced using the heuristic assumption that the system is semiregular. 
Our bounds are conditional to Conjecture~\ref{conj:EGH}, an extensively-studied conjecture from commutative algebra. However, they are proven and do not rely on any heuristic assumption. In particular, they apply to the situations where no semiregular sequence exists. This is the case for certain choices of the parameters and over a given field. 
On the other hand, since they are sharp, our bounds give us an idea of what security one can hope to achieve for a system with given parameters. 

The paper is structured as follows. Section 2 contains some useful preliminaries. In Section 3 we briefly review cryptographic semiregular sequences and semiregular sequences, and we introduce the family of polynomial systems that are regular in a given degree $D$. These are the systems that we will focus on in this paper. In Section 4, after recalling some results from commutative algebra, we establish bounds on the degree of regularity and the solving degree of a polynomial system that is regular in degree $D$. If the field size is smaller than the degree $D$ of the equations, then we establish bounds on the degree of regularity and the solving degree of the polynomial system to which we add the field equations. The bounds appear in Theorem~\ref{teor:DD:LPP}, Corollary~\ref{cor:DD:salizzoni}, Proposition~\ref{prop:D>q}, and Theorem~\ref{teor:D>q}.

\section{Preliminaries}

Let $\mathbb{F}$ be an arbitrary field and let $\FF_q$ be the finite field of cardinality $q$. Let $R=\mathbb{F}[x_1,\ldots,x_n]$ be the ring of polynomials in $n$ variables with coefficients in the field $\mathbb{F}$. Even though our main motivation comes from cryptography, most of the results in this paper do not depend on the ground field and will be stated and proven over an arbitrary field. 

We start by recalling how a typical multivariate one-way function is constructed. Let $\FF=\FF_q$, $f_1,\ldots, f_m \in R$, and consider the evaluation map
$$\begin{array}{cccc}
	\mathcal{F}: & \mathbb{F}_q^n & \rightarrow & \mathbb{F}_q^m \\
	& \alpha=(\alpha_1,\ldots,\alpha_n) & \mapsto & (f_1(\alpha_1,\ldots,\alpha_n),\ldots,f_m(\alpha_1,\ldots,\alpha_n))
	\end{array}$$
To hide the structure of $\mathcal{F}$, one composes it with two random invertible linear maps $S:\FF_q^n \rightarrow \FF_q^n$ and $T:\FF_q^m \rightarrow \FF_q^m$. One obtains $\mathcal{P}=T \circ \mathcal{F} \circ S$, a set of $m$ polynomials $p_1,\ldots,p_m$ in $n$ variables over~$\FF_q$. The public key of the multivariate scheme is $\mathcal{P}=(p_1,\ldots,p_m)$ and the private key is $\{\mathcal{F},S,T\}$. The trapdoor consists of constructing $\mathcal{F}$ such that $\mathcal{F}^{-1}$ is efficiently computable. Notice that $\mathcal{P}$ should be hard to invert without the knowledge of $S,T$, in particular it should be hard to recover the structure of $\mathcal{F}$ from $\mathcal{P}$. 

Finding inverse images with respect to $\mathcal{P}$ corresponds to computing solutions of multivariate polynomial systems. It is well-known that this problem can be solved by computing a Gr\"obner basis of the system $\mathcal{P}$, we refer the interested reader to~\cite{CG20} for more detail.
The first algorithm for computing Gr\"obner bases appeared in the doctoral thesis of Buchberger~\cite{B65}.
Modern algorithms for computing Gr\"obner bases are based on liner algebra and are often more efficient than Buchberger's. Examples of linear-algebra-based algorithms are $F_4$~\cite{Fau99}, $F_5$~\cite{Fau02}, the XL Algorithm~\cite{CKPS00}, MutantXL~\cite{BDMM09}, and their variants. All these algorithms compute the reduced row echelon form of the Macaulay matrix of the system in a given degree, for one or more degrees. 

We now define Macaulay matrices, see also~\cite[Section~1.2]{BFS15}. 
Fix a term order on $R=\mathbb{F}[x_1,\ldots,x_n]$ and let $\mathcal{F}=\{f_1,\dots,f_m\}\subseteq R$ be a system of homogeneous polynomials.
The columns of the {\bf homogeneous Macaulay matrix} $M_d$ of $\mathcal{F}$ are labeled by the monomials of degree $d$ in $R$ and arranged in decreasing order. The rows of $M_d$ are labeled by polynomials of the form $m_{i,j}f_j$, where $m_{i,j}\in R$ is a monomial such that $\deg(m_{i,j}f_j)=d$.
The entry of $M_d$ in position $(i,j)$ is the coefficient of the monomial of column $j$ in the polynomial corresponding to the $i$-th row.

Let $f_1,\dots,f_m$ be arbitrary polynomials (not necessarily homogeneous).
The columns of the {\bf Macaulay matrix} $M_{\leq d}$ of $\mathcal{F}$ are labeled by the monomials of $R$ of degree $\leq d$, arranged in decreasing order. The rows of $M_{\leq d}$ correspond to polynomials of the form $m_{i,j}f_j$, where $m_{i,j}\in R$ is a monomial such that $\deg(m_{i,j}f_j)\leq d$. The entries of $M_{\leq d}$ are defined as in the homogeneous case. 
The rationale behind the use of homogeneous Macaulay matrices for homogeneous systems is that, for a homogeneous system, the Macaulay matrix $M_{\leq d}$ is a block matrix with blocks $M_d,\ldots,M_0$. 

In order to compute a Gr\"obner basis of the system $\mathcal{F}$, one computes the reduced row echelon form of the corresponding Macaulay matrix in one or more degrees. If the system is homogeneous, one uses homogeneous Macaulay matrices. For a sufficiently large degree, this produces a reduced Gr\"obner basis with respect to the chosen order. Some algorithms, such as MutantXL, use a variation called {\bf mutant strategy} in the non-homogeneous case: If row reduction of the Macaulay matrix $M_{\leq d}$ produces new polynomials $g_1,\ldots,g_{\ell}$ of degree strictly smaller than $d$, one appends to the reduction of $M_{\leq d}$ the polynomials $m_{i,j}g_j$, for every $j$ and every monomial $m_{i,j} \in R$ such that $\deg(m_{i,j}g_j) \leq d$. Then one computes the reduced row echelon form again. Throughout the paper, we refer to algorithms that employ the mutant strategy as {\bf mutant algorithms} and to the others as {\bf standard algorithms}. For a description of the basic version of these algorithms, we refer to~\cite[Section~3.1]{CG20} and~\cite[Section~1]{CG21}.

The complexity of computing the reduced row echelon form of the Macaulay matrices $M_{\leq d}$ and $M_d$ depends on their size, hence on the degree $d$. This motivates the next definition.

\begin{definition}\label{def:solvingdegree}
Let $\mathcal{F}=\{f_1,\ldots,f_m\}\subseteq R$ and let $\tau$ be a term order on $R$. The {\bf solving degree} of $\mathcal{F}$ with respect to $\tau$ is the least degree $d$ such that Gaussian elimination on the Macaulay matrix $M_{\leq d}$ produces a $\tau$-Gr\"obner basis of $\mathcal{F}$. We denote by $\solvdeg_{\tau}^s(\mathcal{F})$ the solving degree of $\mathcal{F}$ with respect to a standard algorithm and by $\solvdeg_{\tau}^m(\mathcal{F})$ the solving degree of $\mathcal{F}$ with respect to a mutant algorithm. When $\tau$ is the degree reverse lexicographic order, we omit the subscript $\tau$.
\end{definition}

It is well-known that the solving degree is not invariant under coordinate change. In addition, the solving degree may depend on the algorithm used to perform the Gr\"obner basis computation. In particular, for mutant algorithms it may be smaller than for standard ones. Finally, the solving degree depends on the choice of a term order on $R$. 

The complexity of linear-algebra-based algorithms is dominated by the cost of computing a degree reverse lexicographic Gr\"obner basis of the system, see~\cite[Sections 2 and 3]{CG20} for more detail. Therefore, an upper bound on the solving degree with respect to the reverse lexicographic order yields an upper bound on the complexity of computing a lexicographic Gr\"obner basis, hence on the complexity of solving the polynomial system. 

Let $I$ be a homogeneous ideal of $R$. For an integer $d\geq0$, we denote by $R_d$ the vector space generated by the degree $d$ monomials of $R$. We denote by $I_d=I\cap R_d$ the $\mathbb{F}$-vector space of homogeneous polynomials of degree $d$ in $I$.
For $g \in R$ a polynomial, we denote by $g^{\topp}$ the homogeneous part of $g$ of largest degree. For example, if $g=x^3+2xy^2-y+1 \in \FF[x,y]$, then $g^{\topp}=x^3+2xy^2$. For a polynomial system $\mathcal{F}=\{f_1,\ldots,f_m\}\subseteq R$, we denote by $\mathcal{F}^{\topp} \subseteq R$ the homogeneous system $\{f_1^{\topp},\ldots, f_m^{\topp}\}$. Up to performing Gaussian elimination on a matrix whose rows correspond to $f_1,\ldots,f_m$, we may suppose that 
\begin{center}
{\bf $f_1^{\topp},\ldots, f_m^{\topp}$ are linearly independent}.
\end{center}  
Throughout the paper, we always make this assumption.

The {\bf degree of regularity} was introduced in~\cite{B04,BFS04}. 

\begin{definition}\label{def:dreg}
Let $\mathcal{F}\subseteq R$ be a polynomial system.
The {\bf degree of regularity} of $\mathcal{F}$ is
$$d_{\reg}(\mathcal{F})=\left\{\begin{array}{cc}
\min\{d\geq 0\mid (\mathcal{F}^{\mathrm{top}})_d=R_d\} & \mbox{ if } \left(\mathcal{F}^{\mathrm{top}}\right)_d=R_d \mbox{ for $d\gg0$} \\
+\infty & \mbox{ otherwise.}
\end{array}\right.$$
\end{definition}

In the cryptographic literature, the degree of regularity is often assumed to be finite and is used as a proxy for the solving degree. This is the case, e.g., in the GeMSS specification documents~\cite{CFMPPR17}. However, this does not always produce reliable estimates, even in the case when the degree of regularity is finite. In fact, there are examples in which the degree of regularity is much smaller than the solving degree, see e.g.~\cite[Examples 3.2 and 3.3]{BDDGMT21}. 
However, a  result by Semaev and Tenti~\cite{ST21,T19} shows that, under suitable assumptions, the solving degree of a system which contains the field equations is at most twice the degree of regularity. 
Due to this result, an upper bound for the degree of regularity yields a proven upper bound for the solving degree of standard algorithms.

\begin{teor}[{\cite[Corollary 3.67]{T19} and~\cite[Theorem 2.1]{ST21}}] \label{teor:tenti:semaev}
Let $\mathcal{F} = \{f_1, \ldots, f_m, x_{1}^{q}-x_{1}, \ldots , x_{n}^{q}-x_{n}\} \subseteq R$ be a polynomial system. If $d_{\reg}(\mathcal{F}) \geq \max\{q, \deg(f_1),\ldots, \deg(f_m)\}$, then 
$$\solvdeg^s(\mathcal{F}) \leq 2d_{\reg}(\mathcal{F}) - 2.$$
\end{teor}

Notice that in most systems of cryptographic interest, the field size and the degrees of the polynomials are relatively small. Therefore, one expects that Theorem~\ref{teor:tenti:semaev} applies to such systems.

Another recent result by Salizzoni~\cite{Sal23} shows that the solving degree of a mutant algorithm is at most the degree of regularity plus one, unless the system contains polynomials of large degree. 

\begin{teor}[{\cite[Proposition 3.10]{Sal23}}]\label{teor:salizzoni}
    Let $\mathcal{F}=\{f_1,\ldots,f_m\}\subseteq R$ be a polynomial system, then 
    $$\solvdeg^m(\mathcal{F}) \leq \max\{d_{\reg}(\mathcal{F}) +1, \deg(f_1), \ldots, \deg(f_m)\}.$$
\end{teor}

Notice that, for both standard and mutant algorithms, it makes sense to assume that 
\begin{equation}\label{eqn:sd}
    \solvdeg(\mathcal{F}) \geq \max\{\deg(f_1), \ldots, \deg(f_m)\}.
\end{equation} 
In fact, any equation of $\mathcal{F}$ of degree greater than the solving degree may be removed from the system without altering the Gr\"obner basis computation. This is of course delicate in practice, as one may not have an a priori estimate on the solving degree that allows them to eliminate equations from the system. 
Notice moreover that most multivariate systems associated to cryptographic schemes have low degree, so (\ref{eqn:sd}) is usually satisfied. 

In some proofs, we consider the Hilbert function and Hilbert series of a quotient of $R$ by a homogeneous ideal. We refer to~\cite[Chapter 4]{BH} for the definitions and basic properties.
Another algebraic invariant connected to the solving degree of a polynomial system is the {\bf Castelnuovo-Mumford regularity}. Throughout the paper, we denote by $\reg(I)$ the Castelnuovo-Mumford regularity of the ideal $I$.
We refer the reader to~\cite[Section 3.4]{CG20} and~\cite{CG21} for its definition and a discussion of its relationship with the other invariants.
Notice that, if $d_{\reg}(\mathcal{F})<+\infty$, then $$d_{\reg}(\mathcal{F}) = \reg(\mathcal{F}^{\topp}).$$

\section{Systems that contain a regular sequence in a given degree} \label{sec:randomness}

Consider a system $\mathcal{F}$ of $m$ equations in $n$ variables. 
As we discussed in the previous section, whenever the degree of regularity of $\mathcal{F}$ is finite, an upper bound on it directly translates into an upper bound on the solving degree of the system, hence on the complexity of computing its solutions. 
In this paper, we prove upper bounds on the degree of regularity of a system $\mathcal{F}$ such that the ideal generated by $\mathcal{F}^{\topp}$ contains a regular sequence of polynomials of degree $D$.  
In the rest of this section, we discuss the property of containing a regular sequence in a given degree and in which sense this is a generic property.
The next definition formalizes the correspondence between homogeneous systems and points in a product of projective spaces.

\begin{definition}\label{defn:openset}
Let $\FF$ be a field, and let $\overline{\FF}$ be its algebraic closure. Denote by $\mathbb{P}^t$ the $t$-dimensional projective space over $\overline{\FF}$.
Let $d_1,\ldots,d_m$ be positive integers. For a homogeneous polynomial $f$, denote by $[f]\in\mathbb{P}^{\binom{n+\deg(f)-1}{n-1}-1}$ the projective point whose coordinates are the coefficients of $f$. A homogeneous system $\mathcal{F}=\{f_1,\ldots,f_m\}\subseteq R$ with $\deg(f_i)=d_i$ for $1\leq i\leq m$ corresponds to a point $([f_1],\ldots,[f_m])\in\mathbb{P}^{\binom{n+d_1-1}{n-1}-1}\times\ldots\times\mathbb{P}^{\binom{n+d_m-1}{n-1}-1}$.
\end{definition}

A {\bf random system} of degrees $d_1,\ldots,d_m$ over a finite field $\mathbb{F}$ is a system of polynomials of the given degrees whose coefficients are chosen uniformly at random in the given field. Similarly, a {\bf random homogeneous system} of degrees $d_1,\ldots,d_m$ is a system of homogeneous polynomials of the given degrees, whose coefficients are chosen uniformly at random. The concept of random system is used extensively in the cryptographic literature, where random instances of structured polynomial systems are often assumed to behave as random systems.

Over an infinite field, a {\bf generic property} in the sense of algebraic geometry is by definition a property which holds over a nonempty Zarisky-open set. 
We refer to \cite[Section 1.3]{CG20} for a short reminder on the Zarisky topology. 
Since every nonempty Zarisky-open set is dense, the property holds for ``almost every system''. Therefore, over an infinite field, genericity is the algebraic-geometric analog of being true with high probability.  
Over a finite field $\FF_q$, however, 
estimating the probability that a system has a given property amounts to estimating the number of $\FF_q$-rational points of the open set that corresponds to the property over the algebraic closure, a task which is often far from trivial.
Nevertheless, if $q$ is large enough (or if we consider a finite extension of $\FF_q$ of large enough cardinality), a random homogeneous system belongs to any given nonempty Zarisky-open set with high probability. For the convenience of the reader, we show how to deduce this from known results. For simplicity, we only state the argument for one projective space, but the same argument can be applied to a product of projective spaces.

\begin{prop}\label{prop:probab}
Let $\mathbb{P}^t$ be the $t$-dimensional projective space over $\bar{\FF_q}$ and let $\mathcal{U}\subseteq\mathbb{P}^t$ be a nonempty Zarisky-open set defined over $\FF_q$. The probability that a uniformly random $\FF_{q^e}$-rational point of $\mathbb{P}^t$ belongs to $\mathcal{U}$ tends to one as $e$ tends to infinity.
\end{prop}

\begin{proof}
Let $\mathcal{X}$ be the set-theoretic complement of $\mathcal{U}$ in $\mathbb{P}^t$. Let $e$ be a positive integer. Since $\mathcal{X}$ is a projective algebraic variety, it follows from \cite[Theorem 3.1]{C16} that the number $|\mathcal{X}(\FF_{q^e})|$ of $\FF_{q^e}$-rational points of $\mathcal{X}$ is bounded from above by $p(q^e)$, where $p(x)$ is a polynomial, whose coefficients depend on the degrees and dimensions of the irreducible components of $\mathcal{X}$. In particular, 
$$\deg p(x)\leq \max\{\dim(\mathcal{X}),2\dim(\mathcal{X})-t\}\leq t-1.$$ Since the number of $\FF_{q^e}$-rational points of $\mathbb{P}^t$ is $\sum_{i=0}^t q^{ei}$, then the probability that an $\FF_{q^e}$-rational belongs to $\mathcal{X}$ is bounded from above by $$\frac{|\mathcal{X}(\FF_{q^e})|}{|\mathbb{P}^t(\FF_{q^e})|}=\frac{p(q^e)}{\sum_{i=0}^t q^{ei}}.$$
Hence, the probability that a uniformly random $\FF_{q^e}$-rational point of $\mathbb{P}^t$ belongs to $\mathcal{X}$ tends to zero as $e$ tends to infinity.
\end{proof}

In this paper, we are interested in regular sequences contained in a given ideal. It is well-known that regular sequences only exist for $m\leq n$, this follows e.g. by applying~\cite[part (3) of the Proposition on pg. 8]{H2} to $(x_1,\ldots,x_n)$.

\begin{definition}\label{defn:reg}
A sequence of polynomials $f_1, \ldots, f_m\in R$ are a {\bf regular sequence} if for all $1\leq i \leq m$ and all $g_i\in R$ such that $g_if_i\in(f_1, \ldots, f_{i-1})$, we have $g_i\in(f_1, \ldots, f_{i-1})$.
\end{definition}

In the cryptographic literature, random polynomial systems are often assumed to be cryptographic semiregular sequences, see e.g.~\cite{B04, BFS03}. This is the case in the cryptoanalysis of several systems, as e.g.~\cite{CFMPPR17}. The next definition appears in~\cite[Definition 3.2.1 and Definition 3.2.4]{B04},~\cite[Definition 5]{BFS04}, and~\cite[Definition 5 and Definition 9]{BFSY05}.

\begin{definition}\label{defn:semireg}\label{defn:semireg:FF2}
Let $\mathcal{F}=\{f_1, \ldots, f_m\}\subseteq R$ be a homogeneous system.

If $\FF\neq\FF_2$, we say that $f_1, \ldots, f_m$ are a {\bf cryptographic semiregular sequence} if for all $1\leq i \leq m$ and all $g_i\in R$ such that $g_if_i\in(f_1, \ldots, f_{i-1})$ and $\deg(g_i f_i) < d_{\reg}(\mathcal{F})$, we have $g_i\in(f_1, \ldots, f_{i-1})$.

If $\FF=\FF_2$, we say that $f_1, \ldots, f_m\in R/(x_1^2,\ldots,x_n^2)$ are a {\bf cryptographic semiregular sequence} if for all $1\leq i \leq m$ and all $g_i\in R/(x_1^2,\ldots,x_n^2)$ such that $g_i f_i \in(f_1, \ldots, f_{i-1})$ and $\deg(g_i f_i) < d_{\reg}(\mathcal{F}\cup\{x_1^2,\ldots,x_n^2\})$, we have $g_i\in(f_1, \ldots, f_i)$.

A sequence of polynomials (not necessarily homogeneous) $f_1, \ldots, f_m$ is a {\bf cryptographic semiregular sequence} if $f_1^{\topp}, \ldots, f_m^{\topp}$ are a cryptographic semiregular sequence. 
\end{definition}

In this paper, we use the expression cryptographic semiregular sequence in order to distinguish the concept of semiregularity used in the cryptographic literature from the concept of semiregularity originally introduced by Pardue~\cite{Par99,Par10}, which inspired it. The original definition by Pardue is given over an infinite field $\FF$. As we are interested in dealing with finite fields, we extend it to an arbitrary field in the natural way.

\begin{definition}
Let $\FF$ be a field, and let $R=\FF[x_1,\ldots,x_n]$. Let $I$ be a homogeneous ideal, and let $A=R/I$. A polynomial $f\in R_d$ is {\bf semiregular} on $A$ if for every $e\geq d$, the vector space map $A_{e-d}\rightarrow A_e$ given by multiplication by $f$ has maximal rank (that is, it is either injective or surjective).
A sequence of homogeneous polynomials $f_1,\ldots,f_m$ is a {\bf semiregular sequence} if $f_i$ is semiregular on $A/(f_1,\ldots,f_{i-1})$ for all $1\leq i \leq m$.
\end{definition}

Notice that a semiregular sequence with $m\leq n$ is regular.
Moreover, it follows from~\cite[Proposition 1]{Par10} that, if $\FF\neq\FF_2$, then a semiregular sequence is also a cryptographic semiregular sequence. The converse does not hold, as shown in~\cite{Par10}, see the example just below~\cite[Proposition 1]{Par10}.

The next proposition presents a simple situation in which cryptographic semiregular sequences and semiregular sequences coincide.

\begin{prop}\label{semireg=crypto}
Let $q>2$ and let $f_1,\ldots,f_{n+1}\in R=\FF_q[x_1,\ldots,x_n]$ be polynomials of degree $d_1,\ldots,d_{n+1}$, respectively. Assume that $f_{1}^{\topp},\ldots,f_n^{\topp}$ is a regular sequence. The sequence $f_1,\ldots,f_{n+1}$ is cryptographic semiregular if and only if the sequence $f_{1}^{\topp},\ldots,f_{n+1}^{\topp}$ is semiregular. In particular, the sequence $x_1^q-x_1,\ldots,x_n^q-x_n,f$ is cryptographic semiregular if and only if the sequence $x_1^q,\ldots,x_n^q,f^{\topp}$ is semiregular.
\end{prop}

\begin{proof}
Let $\mathcal{F}=\{f_1,\ldots,f_{n+1}\}$.
It suffices to show that if the sequence $f_1,\ldots,f_{n+1}$ is cryptographic semiregular, then the sequence $f_{1}^{\topp},\ldots,f_{n+1}^{\topp}$ is semiregular. Since $f_{1}^{\topp},\ldots,f_n^{\topp}$ is a regular sequence, it suffices to show that $f_{n+1}^{\topp}$ is semiregular on $\FF_q[x_1,\ldots,x_n]/(f_{1}^{\topp},\ldots,f_n^{\topp})$.
If $f_1,\ldots,f_{n+1}$ is cryptographic semiregular, then by~\cite[Proposition 6]{BFSY05} the Hilbert series of $R/(\mathcal{F}^{\topp})$ is $[(1-z^{d_{n+1}})\prod_{i=1}^n(1+z+\ldots+z^{d_i-1})]$. Here, for $p(z)=\sum_{i=0}^{\infty} p_iz^i\in\mathbb{Z}[[z]]$, we denote by $\delta(p)=\min\{i\geq 0\mid p_i\leq 0\}-1$ and define $[p(z)]=\sum_{i=0}^{\delta(p)} p_iz^i$. Then $f_{n+1}^{\topp}$ is semiregular on $\FF_q[x_1,\ldots,x_n]/(f_{1}^{\topp},\ldots,f_n^{\topp})$ and $\mathcal{F}^{\topp}$ is a semiregular sequence by~\cite[Proposition 1]{Par10}.
\end{proof}

Pardue in~\cite{Par99,Par10} shows that Fr\"oberg's Conjecture~\cite{Fr85}, a conjecture which has attracted a lot of attention within the commutative algebra community and is widely believed to hold, is equivalent to the following
\begin{conjecture}\label{conj:Froeberg}
Let $\FF$ be an infinite field.
A generic sequence of polynomials of degrees $d_1,\ldots,d_m$ in $R=\FF[x_1,\ldots,x_n]$ is semiregular.
\end{conjecture}

In other words, Fr\"oberg's Conjecture is equivalent to the statement that the set of semiregular sequences of polynomials of given degrees contains a dense Zarisky-open set.
If the conjecture is true, then an argument similar to that of Proposition \ref{prop:probab} shows that a sequence of polynomials of given degrees is semiregular with high probability, provided that the ground field has large enough cardinality. It follows that, if $\FF=\FF_q$ with $q\gg 0$ and Fr\"oberg's Conjecture holds, then a sequence of polynomials of given degrees is a cryptographic semiregular sequence with high probability. In addition, for $q\gg 0$ most cryptographic semiregular sequences are also semiregular sequences, as the set of semiregular sequences conjecturally contains a dense open set.

In~\cite[Section 3.2]{B04}, Bardet conjectures that a sequence of polynomials with coefficients in $\mathbb{F}_2$ is cryptographic semiregular with high probability. This conjecture is motivated by experimental evidence, see also~\cite[Conjecture 2]{BFS03}. The conjecture was later disproved by Hodges, Molina, and Schlather, who in~\cite{HMS17} prove that there are choices of the parameters for which no cryptographic semiregular sequence exists over $\FF_2$. This is the case, e.g., for $m=1$ and $n>3d_1$. 

In this paper, we study the Zariski-open set consisting of systems of $m$ polynomials in $n$ variables which generate an ideal that contains a regular sequence of $n$ polynomials in a given degree $D$. 

\begin{definition}\label{defn:regD}
Let $\FF$ be a field and let $\overline{\FF}$ be its algebraic closure. Denote by $\mathbb{P}^t$ the $t$-dimensional projective space over $\overline{\FF}$.
Fix $m\geq n\geq 1$ and $1\leq d_1\leq\ldots\leq d_m$.
Let $\mathcal{U}_D\subseteq\mathbb{P}^{\binom{n+d_1-1}{n-1}-1}\times\ldots\times\mathbb{P}^{\binom{n+d_m-1}{n-1}-1}$ be the set of $([g_1],\ldots,[g_m])\in\mathbb{P}^{\binom{n+d_1-1}{n-1}-1}\times\ldots\times\mathbb{P}^{\binom{n+d_m-1}{n-1}-1}$ such that the ideal $(g_1,\ldots,g_m)\subseteq\overline{\FF}[x_1,\ldots,x_n]$ contains a regular sequence of $n$ polynomials of degree $D$. 
We say that the system $\mathcal{F}$ is {\bf regular in degree $D$} if $([f_1^{\topp}],\ldots,[f_m^{\topp}])\in \mathcal{U}_D$.
\end{definition}

Fix $m\geq n\geq 1$, $1\leq d_1\leq\ldots\leq d_m$, and $\mathbb{F}$ a field. 
If $D<d_n$, then $\mathcal{U}_D=\emptyset$. In fact, let $\mathcal{F}=\{f_1,\ldots,f_m\}\subseteq R$ with $\deg(f_i)=d_i$. Then $(f_1^{\topp},\ldots,f_m^{\topp})_D=(f_1^{\topp},\ldots,f_{n-1}^{\topp})_{D}$ cannot contain a regular sequence of $n$ polynomials. Therefore, the system $\mathcal{F}$ can be regular in degree $D$, only if $D\geq d_n$.

The next proposition provides some useful facts connected to the property of being regular in a given degree. The statement follows from standard facts from commutative algebra. Recall that a system defined over a field $\FF$ is {\bf zero-dimensional} if it has finitely many solutions over $\overline{\FF}$.

\begin{prop}\label{prop:regD}
Let $\mathbb{F}$ be a field, and let $m\geq n\geq 1$ and $1\leq d_1\leq\ldots\leq d_m$. Let $\mathcal{F}=\{f_1,\ldots,f_m\}\subseteq R$ with $\deg(f_i)=d_i$. The following are equivalent:
\begin{enumerate}
\item $d_{\reg}(\mathcal{F})<+\infty$,
\item $\mathcal{F}^{\topp}$ is zero-dimensional,
\item $\mathcal{F}^{\topp}\in\mathcal{U}_D$ for some $D\geq 0$,
\item $\mathcal{F}^{\topp}\in\mathcal{U}_D$ for all $D\geq d_m$.
\end{enumerate}
In particular, if $d_1=\ldots=d_m=D$ and $\mathcal{F}^{\topp}$ is zero-dimensional, then $\mathcal{F}$ is regular in degree $D$. 
\end{prop}

\begin{proof}
1. $\Leftrightarrow$ 2. The finiteness of the degree of regularity can be restated as the fact that the Hilbert series of $R/(\mathcal{F}^{\topp})$ is zero in degree $d\gg 0$. Since the denominator of the Hilbert series is $(1-z)^{\dim(R/(\mathcal{F}^{\topp}))}$, where $\dim(R/(\mathcal{F}^{\topp}))$ is the Krull dimension of $R/(\mathcal{F}^{\topp})$, this is equivalent to the statement that $\mathcal{F}^{\topp}$ is zero-dimensional. \\
1. $\Rightarrow$ 3. Let $D=d_{\reg}(\mathcal{F})$. Then $(\mathcal{F}^{\topp})_D=R_D$, so it contains the regular sequence $x_1^D,\ldots,x_n^D$. \\
3. $\Rightarrow$ 4. Denote $(\mathcal{F}^{\topp})\subseteq\overline{\FF}[x_1,\ldots,x_n]$ by $I$ for brevity. By assumption, there exists a regular sequence $g_1,\ldots,g_n\in I_D$. 
Let $P_1,\ldots,P_s$ be the associated primes of $(g_1,\ldots,g_{n-1})$. If $I_{d_m}\subseteq (P_1)_{d_m}\cup\ldots\cup (P_s)_{d_m}$, then $I_{d_m}\subseteq (P_i)_{d_m}$ for some $i$ by \cite[Theorem 10.4 (2)]{H1}, therefore $I_D=I_{d_m}(x_1,\ldots,x_n)_{D-d_m}\subseteq (P_i)_D$. This contradicts the assumption that $g_n\in I_D\setminus(P_1\cup\ldots\cup P_s)$. Therefore there exists $h_n$ such that $g_1,\ldots,g_{n-1},h_n$ are a regular sequence. By \cite[Proposition on pg. 1]{H2} every permutation of a regular sequence is a regular sequence, in particular $g_1,\ldots,g_{n-2},h_{n},g_{n-1}$ are a regular sequence. By the same reasoning as above, $g_{n-1}$ can be replaced with $h_{n-1}\in I_{d_m}$ so that $g_1,\ldots,g_{n-2},h_{n},h_{n-1}$ are a regular sequence. Iterating this argument, one construct a regular sequence $h_1,\ldots,h_n\in I_{d_m}$.\\
4. $\Rightarrow$ 1. If $\mathcal{F}^{\topp}\in\mathcal{U}_{d_m}$, then the Hilbert series of $R/(\mathcal{F}^{\topp})$ is bounded from above by $(1-z^{d_m})^n/(1-z)^{n}=(1+\cdots+z^{d_m-1})^n$, that is, by the Hilbert series of a regular sequence of $n$ polynomials of degree $d_m$. Since the latter is a polynomial, the Hilbert series of $R/(\mathcal{F}^{\topp})$ is a polynomial, hence $d_{\reg}(\mathcal{F})<+\infty$.
\end{proof}

\begin{remark}
Notice that for the equivalence of 3. and 4. it is essential that the field over which we look for a regular sequence is infinite. In fact, over a finite field, it is possible that a vector space is a nontrivial union of vector spaces. For example, over $\mathbb{F}_2$ \begin{equation}\label{eqn:deg23}
\langle x^2,y^2\rangle=(x)_2\cup(y)_2\cup(x+y)_2.
\end{equation} 
Coherently, the ideal $(x^2,y^2)\subseteq R=\FF_2[x,y]$ is not contained in $(x)\cup(y)\cup(x+y)$, since $x^3+x^2y+y^3\in (x^2,y^2)\not\subseteq(x)\cup(y)\cup(x+y)$. However, by (\ref{eqn:deg23}) there is no $f\in (x^2,y^2)_2$ such that $f\not\in(x)\cup(y)\cup(x+y)$.
\end{remark}

It follows from Proposition~\ref{prop:regD} that, if $m=n$, a system $\mathcal{F}$ has finite degree of regularity
if and only if $\mathcal{F}^{\topp}$ is a regular sequence. Since this case is well-studied, in the sequel we often focus on the case where $m>n$. 

Notice moreover that, if $\mathcal{F}$ is cryptographic semiregular, then it satisfies the equivalent conditions of Proposition~\ref{prop:regD}. In other words, every cryptographic semiregular system is regular in degree $D$ for every $D\geq d_m$ (and possibly in smaller degree). However, the converse does not hold.

The statement that every cryptographic semiregular sequence is regular in some degree also implies that being regular in degree $D$ is a generic property, conditionally to Conjecture~\ref{conj:Froeberg}. More precisely, Proposition~\ref{prop:regD} implies that, if $D\geq d_n$, then $\mathcal{U}_D$ contains a dense open set of $\mathbb{P}^{\binom{n+d_1-1}{n-1}-1}\times\ldots\times\mathbb{P}^{\binom{n+d_m-1}{n-1}-1}$. 

Even over a finite field, one advantage of considering systems which are regular in degree $D$ is that, unlike cryptographic semiregular systems~\cite{HMS17}, they exist over every field and for any choice of $m\geq n$ and $D\geq d_n$. In other words, if the set $\mathcal{U}_D$ is nonempty over $\overline{\mathbb{F}_q}$, then it always contains $\mathbb{F}_q$-rational points, as the following example shows. 

\begin{example}\label{ex:dense}
The set $\mathcal{U}_D$ contains $\mathbb{F}_q$-rational points for any choice of $n,m,q$ and $1\leq d_1\leq\ldots\leq d_m$ with $n\leq m$ and $D\geq d_n$. This follows from the existence of homogeneous regular sequences of any given degrees in $\mathbb{F}_q[x_1,\ldots,x_n]$.
Some homogeneous regular sequences in $\mathbb{F}_q[x_1,\ldots,x_n]$ of degrees $d_1,\ldots,d_n$ are for example $$x_1^{d_1}+g_1,x_2^{d_2}+g_2,\ldots,x_n^{d_n}+g_n,$$ where $g_i\in\mathbb{F}_q[x_{i+1},\ldots,x_n]_{d_i}$. 
It follows that the system $$\{x_1^{d_1}+g_1,x_2^{d_2}+g_2,\ldots,x_n^{d_n}+g_n,f_{n+1},\ldots,f_m\}$$ is regular in degree $D$ for every $D\geq d_n$ and for any $f_{n+1},\ldots,f_m$ of the appropriate degrees.
\end{example}

We stress that Proposition~\ref{prop:regD} applies to most zero-dimensional systems, since for most zero-dimensional systems over an algebraically closed field, the top degree part of the system $\mathcal{F}^{\topp}$ is also zero-dimensional. This statement can be formalized using the algebraic geometric language of dense open sets, as $\mathcal{U}_D$ is dense by Example~\ref{ex:dense}. Nevertheless, there exist zero-dimensional systems $\mathcal{F}$ such that $\mathcal{F}^{\topp}$ is not zero-dimensional. 

\begin{example}\label{ex:zerodimnotinU}
Let $f_1=x^3$, $f_2=xy^2+y^2$, $f_3=x^2y$. The system is zero-dimensional, as $\sqrt{(f_1,f_2,f_3)}=(x,y)$. However, the system is not regular in degree $3$, since $(f_1^{\topp},f_2^{\topp},f_3^{\topp})=(x^3,x^2y,xy^2)$ does not contain a regular sequence in degree $3$. In fact, $(f_1^{\topp},f_2^{\topp},f_3^{\topp})$ is not even zero-dimensional, as $\sqrt{(f_1^{\topp},f_2^{\topp},f_3^{\topp})}=(x)$.
\end{example}

\begin{example}
The system $\mathcal{F}\cup\{x_1^q-x_1, \ldots,x_n^q-x_n\}$ is regular in degree $q$, for any system $\mathcal{F}$. In fact, $(f_1^{\topp},\ldots,f_m^{\topp},x_1^q,\ldots,x_n^q)_q$ contains the regular sequence $x_1^q,\ldots,x_n^q$.
\end{example}

\section{Bounds on the degree of regularity} \label{sec:result}

Throughout the section, $\mathcal{F}=\{f_1,\ldots,f_m\}\subseteq R$ denotes a system of $m$ polynomials in $n$ variables. We always assume that $f_1^{\topp},\ldots,f_m^{\topp}$ are linearly independent and $m\geq n$. Notice that 
$$d_{\reg}(\mathcal{F})\geq d_n$$ by definition. 
In this section, we establish an upper bound for the degree of regularity of a system $\mathcal{F}$ of polynomials of equal degree $D$, under the assumption that $\mathcal{F}$ is regular in degree $D$. 
We also prove an upper bound for the degree of regularity of a system $\mathcal{F}$ of polynomials of equal degree $D$, for $D\geq q$. In combination with Theorem~\ref{teor:tenti:semaev} and Theorem~\ref{teor:salizzoni}, these results provide us with upper bounds for the solving degree of $\mathcal{F}$. 

\begin{remark}
Let $\mathcal{F}$ be a system of polynomials with $f_1^{\topp},\ldots,f_m^{\topp}$ linearly independent of degree $D$. Let $S$ and $T$ be invertible square matrices of size $n$ and $m$. 
Then $(T\circ \mathcal{F}\circ S)^{\topp}=T\circ\mathcal{F}^{\topp}\circ S$, so $\mathcal{F}$ is regular in degree $D$ if and only if $T\circ\mathcal{F}\circ S$ is regular in degree $D$. In other words, when deciding whether a system of 
polynomials of equal degree $D$ is regular in degree $D$, one can ignore the random linear transformations $S$ and $T$ used to disguise the internal system $\mathcal{F}$. 
\end{remark}

In~\cite[Section 4]{BDDGMT21}, the authors provide an upper bound for the degree of regularity of a system of quadratic polynomials that contains a regular sequence. In this section, we follow the same basic approach and extend it to systems of polynomials of the same degree and systems of polynomials of the same degree to which one adds the field equations. The cases we treat in this paper are technically more challenging and require the use of more sophisticated results from commutative algebra. 

We start by introducing the family of lex-segment ideals. A conjecture by Eisenbud, Green, and Harris will allow us to reduce to these ideals, when estimating the Castelnuovo-Mumford regularity of ideals generated by systems which are regular in degree $D$, for some $D>0$. Throughout the section, we fix the lexicographic term order on $R$ with $x_1>x_2>\ldots>x_n$.

\begin{definition}\label{def:LPP}
A monomial ideal $I\subseteq R$ is a {\bf lex-segment ideal} if it has the property that if $u,v\in R$ are monomials of the same degree such that $u\geq_{\lex} v$ and $v\in I$, then $u\in I$.  

Let $C$ and $c_1 \leq\ldots\leq c_n$ be non negative integers. An ideal $\mathcal{L}\subseteq R$ is a {\bf $(c_1,\ldots,c_n;C)$-LexPlusPowers (LPP) ideal} if $\mathcal{L}=(x_1^{c_1},\ldots,x_n^{c_n})+L$, where $L$ is a lex-segment ideal generated in degree $C$.
\end{definition}

\begin{notation}\label{def:I:LPP}
Let $I \subseteq R$ be a homogeneous ideal containing a regular sequence of polynomials of degrees $c_1 \leq \ldots \leq c_n$. For each $C\geq 0$, we denote by $\LPP(I;c_1,\ldots,c_n;C)$ the $(c_1,\ldots,c_n;C)$-LPP ideal $\mathcal{L}=(x_1^{c_1},\ldots,x_n^{c_n})+L$ such that $\dim(I_C)=\dim(\mathcal{L}_C)$. We make $L$ unique by choosing the largest lex-segment ideal generated in degree $C$ for which the equality $\mathcal{L}=(x_1^{c_1},\ldots,x_n^{c_n})+L$ holds. 
\end{notation}

\begin{example}
Let $I\subseteq\FF[x,y,z]$ be generated by a homogeneous regular sequence of polynomials of degrees $1,3,4$. The $(1,3,4;3)$-LPP ideal $\mathcal{L}$ such that $\dim(I_3)=7=\dim(\mathcal{L}_3)$ is 
$\mathcal{L}=(x,y^3,z^4)$.

For any lex-segment ideal $L$ generated in degree $3$ with $\dim(L_3)\leq 7$, one has $L\subseteq\mathcal{L}$. Hence, one may write $$\mathcal{L}=(x,y^3,z^4)+L$$ as in Definition~\ref{def:LPP} and choose, e.g., $L= (x^3,x^2y,x^2z, xy^2,xyz,xz^2)$, or $L=(x^3)$.
However, Notation~\ref{def:I:LPP} prescribes that we choose the largest $L$ with respect to containment, i.e., $$L=(x^3,x^2y,x^2z, xy^2,xyz,xz^2,y^3)$$ is the lex-segment ideal generated by the first $7$ cubic monomials in the lexicographic order.
\end{example}

The next conjecture appears as~\cite[Conjecture ($V_m$)]{EGH93}. See also~\cite[Conjecture 2.6]{R04} for an equivalent formulation, that is closer to the formulation that we adopt here. The conjecture was settled in several cases and it is widely believed to hold within the commutative algebra community. For an introduction to the conjecture and an excellent survey of known cases, we refer the interested reader to~\cite{CDS}. Here we state it in a weak form, which is what we need in the sequel. 

\begin{conjecture}[Eisenbud-Green-Harris Conjecture]\label{conj:EGH}
Let $I \subseteq R$ be a homogeneous ideal containing a regular sequence of polynomials of degrees $c_1 \leq \ldots \leq c_n$. Then $$\reg(I) \leq \reg(\LPP(I;c_1,\ldots, c_n;C)) $$ for all $C\geq c_n$.
\end{conjecture}

The conjecture is motivated by the following observation.
A classical result by Macaulay shows that, among all homogeneous ideals with the same Hilbert function in a given degree $C$, the lex-segment generated in degree $C$ has the smallest Hilbert function in the next degree. If one restricts to ideals that contain a regular sequence in given degrees, the natural analog of the lex-segment is the LPP ideal from Notation~\ref{def:I:LPP}. For ideals that contain a regular sequence of maximum cardinality, minimum growth of the Hilbert function corresponds to maximum Castelnuovo-Mumford regularity.
In order to compute the degree of regularity of an LPP ideal, we use the following result by Caviglia and De Stefani.

\begin{prop}[{\cite[Lemma 2.3]{CD21}}]\label{prop:reg:LPP}
Let $c_1\leq \ldots \leq c_n$ and $2\leq C \leq \sum_{i=1}^{n} (c_i-1)$. Let $\mathcal{L}=(x_1^{c_1},\ldots,x_n^{c_n})+L$ be a $(c_1,\ldots,c_n;C)$-LPP ideal, and assume that $\mathcal{L} \neq (x_1^{c_1},\ldots,x_n^{c_n})$. Let $u=x_k^{t_k}v$, with $t_k \neq 0$ and $v\in\mathbb{F}[x_{k+1},\ldots,x_n]$, be the smallest monomial with respect to the lexicographic order which belongs to $L$ and has degree $C$. Then $$\reg(\mathcal{L})=t_k+\sum_{i=k+1}^{n}(c_i-1).$$
\end{prop}

Our first result is an explicit upper bound for the degree of regularity of a system of polynomials that is regular in degree $D$. In order to make the proof more readable, we introduce the following notation.

\begin{notation}
If $u \in R_D$ is a monomial that only involves the variables $x_{k},\ldots,x_n$ and has degree $a$ in $x_k$, we say that $u$ is a $(D,k,a)$-\emph{type} monomial. 
\end{notation}

\subsection{Systems which do not contain the field equations}

In the next theorem, we provide an explicit upper bound for the degree of regularity of a system that is regular in degree $D$. The result is conditional on Conjecture~\ref{conj:EGH}. Notice that the result can be applied to practically any cryptographic system of equations of the same degree. In fact, even if the system of interest is not zero-dimensional, one can assign random values to some of the variables in order to obtain a zero-dimensional system, which is then solved for the remaining variables. 

\begin{teor} \label{teor:DD:LPP}
Assume that Conjecture~\ref{conj:EGH} holds.
Let $\mathcal{F}=\{f_1,\ldots,f_m\}\subseteq R$ be a polynomial system and assume that $f_1^{\topp},\ldots,f_m^{\topp}$ are linearly independent of degree~$D$. 
If~$m=n$, then let $k=0$, $t=D-1$.
Else, let $1\leq k\leq n-1$ and $1\leq t\leq D-1$ be such that $m-n$ belongs to the interval $(\sigma_{k,t-1},\sigma_{k,t}]$, where
$$\sigma_{k,t}=\sum_{i=1}^{k}\sum_{j=1}^{D-1} \binom{n-i-1+j}{j} - \sum_{j=t+1}^{D-1} \binom{n-k-1+j}{j} \;\;\mbox{ for }\; 0\leq t\leq D-1.$$
If $d_{\reg}(\mathcal{F})<+\infty$, then $$d_{\reg}(\mathcal{F})\leq (D-t)+(n-k)(D-1).$$
\end{teor}

\begin{proof}
If $m=n$, then $\mathcal{F}^{\topp}$ is a regular sequence of $n$ polynomials of degree $D$, hence $$d_{\reg}(\mathcal{F})=n(D-1)+1.$$ Suppose therefore that $m>n$ and let $J$ be the ideal generated by $\mathcal{F}^{\topp}$. Consider the lexicographic order on $R$ and let $$\mathcal{L}=\LPP(J;D,\ldots,D;D)=(x_1^D,\ldots,x_n^D)+L,$$ be the LPP ideal with $L$ the largest lex-segment ideal generated in degree $D$ such that $\dim(\mathcal{L}_D)=m.$ Since both $f_1^{\topp},\ldots,f_m^{\topp}$ and $x_1^D,\ldots,x_n^D$ are linearly independent, we have
$$\dim\left(\mathcal{L}_D/\langle x_1^D,\ldots,x_n^D\rangle\right)=m-n.$$ 

In order to apply Proposition~\ref{prop:reg:LPP}, we need to determine the value of $k$ and $t_k$ for the monomial in position $m-n$ in the ordered list of degree $D$ monomials in $\FF[x_1,\ldots,x_n]/(x_1^D,\ldots,x_n^D)$. As we will show, these values are the unique $k$ and $t_k$ such that $m-n$ belongs to the interval $(\sigma_{k,D-t_k-1},\sigma_{k,D-t_k}]$. 

For $1\leq k\leq n-1$ and $0\leq t\leq D-1$, the number of $(D,k,D-t)$-type monomials is 
$$\dim(\FF[x_{k+1},\ldots,x_{n}]_{t})=\binom{n-k-1+t}{t},$$
hence
\begin{equation*} 
\begin{split}
\sigma_{k,t} & = \sum_{i=1}^{k-1} \sum_{j=1}^{D-1} \binom{n-i-1+j}{j}  + \sum_{j=1}^{t} \binom{n-k-1+j}{j}  \\
 & = \sum_{i=1}^{k-1} \sum_{j=1}^{D-1} \dim ( \FF[x_{i+1},\ldots,x_{n}])_{j}  + \sum_{j=1}^{t}\dim ( \FF[x_{k+1},\ldots,x_{n}])_{j} 
\end{split}
\end{equation*}
is the number of degree $D$ monomials in $\FF[x_1,\ldots,x_n]$ different from $x_1^D,\ldots,x_n^D$ and bigger than or equal to $x_k^{D-t}x_n^t$, the smallest $(D,k,D-t)$-type monomial. 
In other words, the monomial $x_k^{D-t}x_n^t$ is in position $\sigma_{k,t}$ in the ordered list of degree $D$ monomials in $\FF[x_1,\ldots,x_n]/(x_1^D,\ldots,x_n^D)$. Notice moreover that $\sigma_{1,0}=0$ and $\sigma_{k,0}=\sigma_{k-1,D-1}$ for $2\leq k\leq n-1$.

If $u$ is the smallest monomial in $\mathcal{L}_D/\langle x_1^D,\ldots,x_n^D\rangle$, then $u$ is in position $m-n$ in the ordered list of degree $D$ monomials in $\FF[x_1,\ldots,x_n]/(x_1^D,\ldots,x_n^D)$. 
If $\sigma_{k,t-1}<m-n\leq\sigma_{k,t}$ for some $1\leq k\leq n-1$ and $1\leq t\leq D-1$, then $u$ is a $(D,k,D-t)$-type monomial. Since $0=\sigma_{1,0}<m-n\leq\dim(\FF[x_1,\ldots,x_n]_D)/(x_1^D,\ldots,x_n^D)_D=\sigma_{n-1,D-1}$, then $m-n$ always belong to one of the intervals above.

The ideal $L$ is generated by the degree $D$ monomials which are greater than or equal to $u$, unless the monomial following $u$ in lexicographic decreasing order is a pure power. In that case, $u=x_kx_n^{D-1}$ for some $1\leq k\leq n-1$, i.e. $m-n=\sigma_{k,D-1}$, and the smallest degree $D$ monomial in $L$ is $x_{k+1}^D$. Then $L$ is generated by the degree $D$ monomials which are greater than or equal to $x_{k+1}^D$, which is a $(D,k+1,D)$-type monomial.
In both situations
$$\reg(\mathcal{L})= (D-t) + (n-k)(D-1)$$
by Proposition~\ref{prop:reg:LPP}. The thesis now follows from observing that 
\begin{equation} \label{eq:dreg<reg}
d_{\reg}(\mathcal{F})=\reg(J)\leq \reg(\mathcal{L}),
\end{equation}
where the inequality follows from Conjecture~\ref{conj:EGH}.
\end{proof}

We stress that the upper bound from Theorem~\ref{teor:DD:LPP} is sharp for all values of $m,n,D$. In fact, it is met by any system $\mathcal{F}$ such that $(f_1^{\topp},\ldots,f_m^{\topp})$ is a $(D,\ldots,D;D)$-LPP ideal.

In the next example we show how to compute the bound from Theorem~\ref{teor:DD:LPP} for concrete choices of the parameters.

\begin{example} 
Let $n=6$ and $D=3$, so  $\mathcal{F}=\{f_1,\ldots,f_m\} \subseteq \FF_q[x_1,\ldots,x_6]_3$. Then $1\leq k\leq 5$ and $0\leq t\leq 2$. The values of $\sigma_{k,t}$ are
\begin{center}
\begin{tabular}{ |c|c|c|c| } 
\hline
$\sigma_{k,t}$ & $t=0$ & $t=1$ & $t=2$ \\ \hline
$k=1$ & 0 & 5 & 20 \\ \hline 
$k=2$ & 20 & 24 & 34 \\  \hline
$k=3$ & 34 & 37 & 43 \\ \hline 
$k=4$ & 43 & 45 & 48 \\ \hline 
$k=5$ & 48 & 49 & 50 \\ \hline
\end{tabular}
\end{center}
If $m=12$, then $m-n=6$ and $\sigma_{1,1}< 6 \leq \sigma_{1,2}$. Hence $k=1$, $t=2$, and
$d_{\reg}(\mathcal{F}) \leq 11.$\\
\noindent
If $m=42$, then $m-n=36$ and $\sigma_{3,0}< 36 \leq \sigma_{3,1}$. Hence $k=3$, $t=1$, and
$d_{\reg}(\mathcal{F}) \leq 8.$\\
\noindent
If $m=54$, then $m-n=48$ and $\sigma_{4,1}< 48 \leq \sigma_{4,2}$. Hence $k=4$, $t=2$, and
$d_{\reg}(\mathcal{F}) \leq 5.$
\end{example}

\begin{remark}
The upper bound of Theorem \ref{teor:DD:LPP} is decreasing as a function of $m$, as one would expect. In particular, as $m-n$ passes from an interval $(\sigma_{k,t-1},\sigma_{k,t}]$ to the next, the upper bound decreases by one. 
The largest value of the bound is obtained in the case $m=n$, which corresponds to $\mathcal{F}^{\topp}$ being a regular sequence. In this case, the value for the bound is well-known and is $n(D-1)+1$. The smallest value for the bound is obtained in the case $m-n=\sigma_{n-1,D-1}$, which corresponds to $\langle\mathcal{F}^{\topp}\rangle+\langle x_1^D,\ldots,x_n^D\rangle=R_D$. In this case, the value of the bound is $D$. 
\end{remark}

In order to provide a concrete example, we apply Theorem~\ref{teor:DD:LPP} to the Cubic Simple Matrix encryption scheme proposed in~\cite{DPW14}. This provides the only proven bound on the degree of regularity of this system, to the best of our knowledge.

\begin{example}
The Cubic Simple Matrix encryption scheme consists of $m=2n$ cubic equations in $n$ variables, with $n$ a square. 
To compute the bound of Theorem~\ref{teor:DD:LPP} for this system, we need to find the unique $1\leq k\leq n-1$ and $t\in\{1,2\}$ such that $m-n=n\in(\sigma_{k,t-1},\sigma_{k,t}]$. Since $\sigma_{1,1}=n-1$ and $\sigma_{1,2}=n-1+{n\choose 2}$, we find $k=1$ and $t=2$. Hence, the Cubic Simple Matrix encryption scheme has $d_{\reg}\leq 2n-1$. 

We are not aware of any other proven bound on the degree of regularity of this system. However, \cite{DPW14} contains a heuristic analysis for $n\gg 0$ and experimental results for small $n$.
In~\cite[Section 4.2]{DPW14} the authors perform a heuristic analysis and conclude that, for $n$ sufficiently large, the complexity of solving the cubic system using the XL Algorithm should be at least that of solving a random quadratic system of $n$ equations in $n$ variables. The degree of regularity of the latter is $n+1$, provided that the top degree part of the equations is a regular sequence. In the same section, they carry on numerical experiments for $n=4,9,16$, from which they find degree of regularity $5,6,7$, respectively.
\end{example}

Combining Theorem~\ref{teor:DD:LPP} and Theorem~\ref{teor:salizzoni}, we obtain the following. A related result concerning the solving degree of a standard algorithm will be proven after Theorem~\ref{teor:D>q}.

\begin{cor}\label{cor:DD:salizzoni}
Let $\mathcal{F}=\{f_1,\ldots,f_m\}\subseteq R$ be a system of polynomials such that $f_1^{\topp},\ldots,f_m^{\topp}$ are linearly independent of degree~$D$. Assume that $d_{\reg}(\mathcal{F})<+\infty$. If $m=n$, then let $k=0$, $t=D-1$.
Else, let $1\leq k\leq n-1$ and $1\leq t\leq D-1$ be such that $m-n$ belongs to the interval $(\sigma_{k,t-1},\sigma_{k,t}]$. 
If Conjecture~\ref{conj:EGH} holds, then 
$$\solvdeg^m(\mathcal{F})\leq D-t+1+(n-k)(D-1).$$ 
\end{cor}

\begin{proof}
The upper bound found in Theorem~\ref{teor:DD:LPP} for the degree of regularity of $\mathcal{F}$ is bigger than or equal to $D$ for all $k$ and $t$. The thesis then follows from Theorem~\ref{teor:salizzoni}.
\end{proof}

\begin{remark}\label{rmk:D<q}
While the estimates of Theorem~\ref{teor:DD:LPP} and  Corollary~\ref{cor:DD:salizzoni} hold for every $D$, they are most relevant for $D<q$. The reason is that, if $D\geq q$, then one should add the field equations to the system before solving it. In fact, for any $D$ and $q$ one has $d_{\reg}(\mathcal{F}\cup\{x_1^q-x_1,\ldots,x_n^q-x_n\})\leq d_{\reg}(\mathcal{F})$, that is, adding the field equations does not make the degree of regularity increase. Moreover, if $D\geq q$, the bound on the solving degree from Theorem~\ref{teor:salizzoni} does not increase when adding the field equations.
\end{remark}

\subsection{Systems which contain the field equations}

We now bound the degree of regularity and the solving degree of a system of the form $\mathcal{F}\cup\{x_1^q-x_1,\ldots,x_n^q-x_n\}$, where $\mathcal{F}$ is a system of $m$ equations of degree $D$ in $n$ variables. This is especially relevant in the case $\mathbb{F}=\mathbb{F}_q$, but the results that we will prove hold (and are stated) over an arbitrary field  $\mathbb{F}$.  
Because of Remark \ref{rmk:D<q}, we focus on the case when $D\geq q$.
After reducing the equations of $\mathcal{F}$ modulo the field equations, they have degree at most $q-1$ in each variable, hence total degree at most $n(q-1)$. Therefore, in the sequel we may assume without loss of generality that $$D\leq n(q-1).$$

First, we observe that one can easily derive a lower bound on the degree of regularity.

\begin{remark}
Let $\mathcal{F}=\{f_1,\ldots,f_m\}$ be a polynomial system of degree $D \leq n(q-1)$. Since $(\mathcal{F}^{\topp}\cup\{x_1^q,\ldots, x_n^q\})_{D-1}=(x_1^q,\ldots, x_n^q)_{D-1}\neq R_{D-1}$, then
$$d_{\reg}(\mathcal{F} \cup \{x_1^q-x_1,\ldots, x_n^q-x_n\})\geq D.$$
\end{remark}

We now want to derive an upper bound. We start by computing the number of linearly independent homogeneous polynomials of degree $D$ as a function of $n,D,q$.

\begin{lemma}\label{rmk:m}
Let $\mathcal{F}=\{f_1,\ldots,f_m\}\subseteq R$ with $\deg(f_1)=\ldots=\deg(f_m)=D$. Assume that $q\leq D\leq n(q-1)$.
If $f_1^{\topp},\ldots,f_m^{\topp}$ are linearly independent of degree $D$, then 
\begin{equation}\label{boundonm}m\leq\sum_{i=0}^{\lfloor \frac{D}{q} \rfloor}(-1)^i \binom{n}{i}\binom{n+D-1-iq}{n-1}.
\end{equation}
Conversely, if inequality (\ref{boundonm}) is satisfied, then the set of systems $\mathcal{F}$ such that that $f_1^{\topp},\ldots,f_m^{\topp}$ are linearly independent modulo $(x_1^q,\ldots,x_n^q)$ is open and nonempty.
\end{lemma}

\begin{proof}
A standard Hilbert function computation shows that the number of homogeneous polynomials in $n$ variables of degree $D$ which are linearly independent modulo $(x_1^q,\ldots,x_n^q)$ is 
$$\dim(R/(x_1^q,\ldots,x_n^q))_D=\sum_{i=0}^{\lfloor \frac{D}{q} \rfloor}(-1)^i \binom{n}{i}\binom{n+D-1-iq}{n-1}.$$ Therefore, if $f_1^{\topp},\ldots,f_m^{\topp}$ are linearly independent, then inequality (\ref{boundonm}) is satisfied. 
The set of systems $\mathcal{F}$ as in the statement is always open, since linear dependence is a rank condition, hence a closed condition.
Suppose then that inequality (\ref{boundonm}) is satisfied. Then there exist $g_1,\ldots,g_m\in (R/(x_1^q,\ldots,x_n^q))_D$ linearly independent, so the set of systems $\mathcal{F}$ in the statement is nonempty. 
\end{proof}

The simplest case is that of very overdetermined systems, more specifically the case when $f_1^{\topp},\ldots,f_m^{\topp}$ are too many to be linearly independent modulo $(x_1^q,\ldots,x_n^q)$. This happens when $m$ is larger than the bound from Lemma~\ref{rmk:m}. It is easy to show that, in such a situation, the degree of regularity is equal to the degree of the equations of the system.

\begin{prop} \label{prop:D>q}
Let $\mathcal{F}=\{f_1,\ldots,f_m\}$ be a polynomial system of degree $D$, where $q\leq D\leq n(q-1)$. If $m\geq \sum_{k=0}^{\lfloor \frac{D}{q} \rfloor}(-1)^k\binom{n}{k}\binom{n+D-1-kq}{n-1}$, then there is a dense open set $\mathcal{W}$ such that, if $\mathcal{F}\in\mathcal{W}$, then $$d_{\reg}(\mathcal{F}\cup\{x_1^q-x_1,\ldots,x_n^q-x_n\})=D.$$
Moreover,
$$\solvdeg^s(\mathcal{F}\cup\{x_1^q-x_1,\ldots,x_n^q-x_n\})\leq 2D-2$$
and
$$\solvdeg^m(\mathcal{F}\cup\{x_1^q-x_1,\ldots,x_n^q-x_n\})\leq D+1.$$
\end{prop}

\begin{proof}
An arbitrary polynomial of degree $D$ in $n$ variables can be identified with its vector of coefficients and hence, up to scalar multiples, with a point of $\mathbb{P}^{{D+n\choose D}-1}$.
Since by assumption $$m\geq\sum_{k=0}^{\lfloor \frac{D}{q} \rfloor}(-1)^k\binom{n}{k}\binom{n+D-1-kq}{n-1}=\dim(R/(x_1^q,\ldots,x_n^q))_D,$$ then there is an open set $\mathcal{W}\subseteq\mathbb{P}^{{D+n\choose D}-1}$ of $m$-tuples of polynomials of degree $D$ such that $$\langle f_1^{\topp},\ldots,f_m^{\topp}\rangle+(x_1^q,\ldots,x_n^q)_D=R_D.$$ Since $(\mathcal{F}^{\topp}\cup\{x_1^q,\ldots,x_n^q\})_{D-1}=(x_1^q,\ldots,x_n^q)_{D-1}\neq R_{D-1}$, then $$d_{\reg}(\mathcal{F}\cup\{x_1^q-x_1,\ldots,x_n^q-x_n\})=D.$$ The rest of the statement now follows from Theorem~\ref{teor:tenti:semaev} and Theorem~\ref{teor:salizzoni}.
\end{proof}

The next theorem yields an upper bound on the degree of regularity of a system of equations of degree $D$ larger than the field size, to which we add the field equations. 
We start with a preparatory lemma, whose proof follows directly from the definition.

\begin{lemma}\label{lemma:interval}
Let $u,v$ be monomials of type $(D,k,a)$ and $(D,h,b)$,  respectively. If $u\geq v$, then either $k<h$ or $k=h$ and $a\geq b$. In particular, if $u,v,w$ are monomials such that $u\geq v\geq w$ and $u$ and $w$ have the same type, then $v$ also has the same type as $u$ and $w$.
\end{lemma}

The next theorem applies, e.g., to a system of equations of the same degree $D\geq 2$ over $\FF_2$. Notice that, by Lemma~\ref{rmk:m}, the assumption that $f_1^{\topp},\ldots,f_m^{\topp}$ are linearly independent modulo $(x_1^q,\ldots,x_n^q)_D$ implies that $m\leq\sum_{i=0}^{\lfloor \frac{D}{q} \rfloor}(-1)^i \binom{n}{i}\binom{n+D-1-iq}{n-1}$.

\begin{teor}\label{teor:D>q}
Assume that Conjecture~\ref{conj:EGH} holds. 
Let $\mathcal{F}=\{f_1,\ldots,f_m\}$ be a polynomial system of degree $D$, with $q\leq D \leq n(q-1)$. Assume that $f_1^{\topp},\ldots,f_m^{\topp}$ are linearly independent modulo $(x_1^q,\ldots,x_n^q)_D$.
Let $1\leq t\leq q-1$ and $1\leq k\leq n-1$ be such that $m$ belongs to the interval $(\varsigma_{k,t-1},\varsigma_{k,t}]$, where 
$$\varsigma_{k,t} = \sum_{i=1}^{k-1}\sum_{j=1}^{q-1} \eta_{i,j} + \sum_{j=1}^{t} \eta_{k,j}$$
and 
$$\eta_{k,t}= \sum_{i=0}^{\left\lfloor \frac{D+t}{q} \right\rfloor -1} (-1)^{i} \binom{n-k}{i} \binom{n-k-1+D-(i+1)q+t}{n-k-1}.$$
Let $$B=q-t+(n-k)(q-1).$$
If one of the following conditions hold: 
\begin{itemize}
\item $m=\varsigma_{k,t}$, $t\neq q-1$ and $D\geq 2q-t-1$, or
\item $m=\varsigma_{k,t}$, $t=q-1$ and $D\geq 2q-1$, 
\end{itemize}
then
$$d_{\reg}(\mathcal{F}\cup\{x_1^q-x_1,\ldots,x_n^q-x_n\})\leq B-1,$$
$$\solvdeg^s(\mathcal{F}\cup\{x_1^q-x_1,\ldots,x_n^q-x_n\})\leq
2(B-2),$$ and  
$$\solvdeg^m(\mathcal{F}\cup\{x_1^q-x_1,\ldots,x_n^q-x_n\})\leq 
B.$$
In any other case
$$d_{\reg}(\mathcal{F}\cup\{x_1^q-x_1,\ldots,x_n^q-x_n\})\leq 
B,$$
$$\solvdeg^s(\mathcal{F}\cup\{x_1^q-x_1,\ldots,x_n^q-x_n\})\leq 
2(B-1),$$
and 
$$\solvdeg^m(\mathcal{F}\cup\{x_1^q-x_1,\ldots,x_n^q-x_n\})\leq
B+1.$$
\end{teor}

\begin{proof}
The idea of the proof is similar to that of Theorem \ref{teor:DD:LPP}. In order to apply Proposition~\ref{prop:reg:LPP}, we need to determine the value of $k$ and $t_k$ for the monomial in position $m$ in the ordered list of degree $D$ monomials in $\FF[x_1,\ldots,x_n]/(x_1^q,\ldots,x_n^q)$. As we will show, these values are the unique $k$ and $t_k$ such that $m$ belongs to the interval $(\varsigma_{k,q-t_k-1},\varsigma_{k,q-t_k}]$. 

We start by proving that $\eta_{k,t}$ is the number of $(D,k,q-t)$-type monomials in $R$ which are linearly independent modulo $(x_{1}^q,\ldots,x_n^q)_D$, that is
\begin{equation}\label{eqn:eta}
\eta_{k,t} = \dim \left[ \frac{\FF[x_{k+1},\ldots, x_n]}{(x_{k+1}^q,\ldots,x_n^q)}  \right]_{D-q+t}.
\end{equation}
Since $x_{k+1}^q,\ldots,x_n^q$ is a regular sequence in $\FF[x_{k+1},\ldots, x_n]$, a standard computation involving Hilbert series yields the explicit formula
$$\dim \left[ \frac{\FF[x_{k+1},\ldots, x_n]}{(x_{k+1}^q,\ldots,x_n^q)}  \right]_{D-q+t} = \sum_{i=0}^{\left\lfloor \frac{D+t}{q} \right\rfloor -1} (-1)^{i} \binom{n-k}{i} \binom{n-k-1+D-(i+1)q+t}{n-k-1},$$
where we notice that $\eta_{k,t}\neq 0$ only if $D+t-q\leq (n-k)(q-1)$. This establishes the equality in (\ref{eqn:eta}). Similarly to the proof of Theorem~\ref{teor:DD:LPP}, we notice that
$$\varsigma_{k,t} = \sum_{i=1}^{k-1} \sum_{j=1}^{q-1} \eta_{i,j} + \sum_{j=1}^{t}\eta_{k,j}$$
is the number of monomials of degree $D$ which do not belong to $(x_1^q,\ldots,x_n^q)$ and are greater than or equal to $x_k^{q-t}x_n^{D-q+t}$, the smallest $(D,k,q-t)$-type monomial.
Moreover, $\varsigma_{k-1,q-1}=\varsigma_{k,0}$ for $2\leq k\leq n-1$.

Let $I \subseteq \FF[x_1,\ldots,x_n]$ be the ideal generated by $\mathcal{F}^{\topp} \cup \{x_1^q,\ldots,x_n^q\}$. 
Let $\mathcal{L}=(x_1^q,\ldots,x_n^q)+L$ be the $(q,\ldots,q;D)$-LPP ideal such that
$$\dim(I_D)=\dim(\mathcal{L}_D)=m+\dim(x_1^q,\ldots,x_n^q)_D.$$
Hence $\mathcal{L}$ is minimally generated by $x_1^q,\ldots,x_n^q$ and $m$ monomials of degree $D$ that do not belong to $(x_1^q,\ldots,x_n^q)$.
Recall that, by assumption, $L$ is the largest lex-segment ideal generated in degree $D$ such that $\mathcal{L}=(x_1^q,\ldots,x_n^q)+L$.

Let $u$ be the smallest degree $D$ monomial in $\mathcal{L}/(x_1^q,\ldots,x_n^q)$. Notice that $u$ is a $(D,k,q-t)$-type monomial, since $m$ belongs to the interval $(\varsigma_{k,t-1},\varsigma_{k,t}]$. Let $v$ be the smallest degree $D$ monomial in $L$. 
If $u$ is the smallest monomial in $(R/(x_1^q,\ldots,x_n^q))_D$, then $u=x_k^{q-t}x_{k+1}^{q-1}\cdots x_n^{q-1}$ and $D=q-t+(n-k)(q-1)$. Moreover, $\mathcal{L}_D=R_D$ and $v=x_n^D$ is a $(D,n,D)$-type monomial.

Assume now that $u$ is not the smallest monomial in $(R/(x_1^q,\ldots,x_n^q))_D$ and let $w$ be the monomial in $(R/(x_1^q,\ldots,x_n^q))_D$ which follows $u$ in decreasing lexicographic order.
Then $u\geq v\geq w$ and $v$ is the degree $D$ monomial next to $w$ in increasing lexicographic order.
If $m\neq \varsigma_{k,t}$, then $w$ has type $(D,k,q-t)$, hence so does $v$ by Lemma \ref{lemma:interval}. 
Suppose therefore that $m=\varsigma_{k,t}$.
Write $D-q+t=(n-\ell)(q-1)+r$, where $0\leq r<q-1$. Notice that $D-q+t<(n-k)(q-1)$, since $u$ is not the smallest monomial in $(R/(x_1^q,\ldots,x_n^q))_D$. Then $\ell>k$. 
In this situation, $u=x_k^{q-t}x_{\ell}^r x_{\ell+1}^{q-1}\cdots x_n^{q-1}$ and $w=x_k^{q-t-1}x_{k+1}^{q-1}\cdots x_{k+n-\ell}^{q-1}x_{k+n-\ell+1}^{r+1}$. Notice that $w$ is a $(D,k,q-t-1)$-type monomial if $t\neq q-1$ and a $(D,k+1,q-1)$-type monomial if $t=q-1$. If $w$ is not the smallest degree $D$ monomial of its type in $R_D$, then $v$ has the same type as $w$. If $w$ is the smallest degree $D$ monomial of its type in $R_D$, then $w=x_k^{q-t-1}x_{k+1}^{D-q+t+1}$ in the case $t\neq q-1$ and $w=x_{k+1}^{q-1}x_{k+2}^{D-q+1}$ in the case $t=q-1$. Since $w\not\in(x_1^q,\ldots,x_n^q)$, this is only possible if $D\leq 2q-t-2$ for $t\neq q-1$ and $D\leq 2q-2$ for $t=q-1$. In this situation, $v$ has type $(D,k,q-t)$ in the case $t\neq q-1$ and type $(D,k+1,q)$ in the case $t=q-1$.

If Conjecture~\ref{conj:EGH} holds, then 
\begin{equation}\label{eqn:conj}
d_{\reg}(\mathcal{F} \cup \{x_1^q-x_1,\ldots, x_n^q-x_n\})\leq \reg(\mathcal{L}).
\end{equation}
Moreover, by Proposition~\ref{prop:reg:LPP}
\begin{equation}\label{eqn:dreg}
\reg(\mathcal{L})=\left\{\begin{array}{ll}
q-t-1 + (n-k)(q-1) & \mbox{if } m=\varsigma_{k,t} \mbox{ and } \\ & \mbox{either } t\neq q-1 \mbox{ and } D\geq 2q-t-1,\\
 & \mbox{or } t=q-1 \mbox{ and } D\geq 2q-1, \\ \\
q-t + (n-k)(q-1) & \mbox{otherwise.}
\end{array}\right.
\end{equation}
The bound on the degree or regularity now follows from (\ref{eqn:conj}) and (\ref{eqn:dreg}).
The bounds on the solving degree follow from the bound on the degree of regularity, Theorem~\ref{teor:tenti:semaev}, and Theorem~\ref{teor:salizzoni}.
\end{proof}

We state as a corollary the case of a system of binary equations. This is a case of special interest and has a simpler formulation than the general case. 

\begin{cor}\label{cor:D>q}
Assume that Conjecture~\ref{conj:EGH} holds. 
Let $\mathcal{F}=\{f_1,\ldots,f_m\}$ be a system of degree $D$ polynomials in $n$ variables over $\FF_2$, with $2\leq D\leq n$. Assume that $f_1^{\topp},\ldots,f_m^{\topp}$ are linearly independent modulo $x_1^2,\ldots,x_n^2$.
Let $1\leq k\leq n-1$ be such that $m$ belongs to the interval $(\varsigma_{k-1},\varsigma_k]$, where 
$$\varsigma_k=\sum_{i=1}^k {n-i\choose D-1}.$$
If $m=\varsigma_{k}$ and and $D\geq 3$, then
$$d_{\reg}(\mathcal{F}\cup\{x_1^2-x_1,\ldots,x_n^2-x_n\})\leq n-k,$$
$$\solvdeg^s(\mathcal{F}\cup\{x_1^2-x_1,\ldots,x_n^2-x_n\})\leq
2(n-k-1),$$ and  
$$\solvdeg^m(\mathcal{F}\cup\{x_1^2-x_1,\ldots,x_n^2-x_n\})\leq 
n-k+1.$$
Else
$$d_{\reg}(\mathcal{F}\cup\{x_1^2-x_1,\ldots,x_n^2-x_n\})\leq 
n-k+1,$$
$$\solvdeg^s(\mathcal{F}\cup\{x_1^2-x_1,\ldots,x_n^2-x_n\})\leq 
2(n-k),$$
and 
$$\solvdeg^m(\mathcal{F}\cup\{x_1^q-x_1,\ldots,x_n^q-x_n\})\leq
n-k+2.$$
\end{cor}

Corollary~\ref{teor:D>q} applies to, e.g., the Unbalanced Oil and Vinegar scheme (UOV) from~\cite[Section 2]{UOV}, which over $\FF_2$ is a system of $m$ quadratic equations in $n>m$ variables, to which one adds the $n$ quadratic field equations. It also applies to the Oil and Vinegar scheme of degree three from~\cite[Section 9]{UOV} and to the Unbalanced Oil, Vinegar and Salt scheme from~\cite[Section 11]{UOV}. Over $\FF_2$, the first is a system of $m$ cubics in $n>m$ variables, to which one adds the $n$ quadratic field equations. The second is a system of $m$ quadratic equations with the Oil and Vinegar structure in $n>m$ variables, to which one adds $r$ random quadratic equations and the $n$ quadratic field equations. 

\begin{example}
Let $\mathcal{F}$ be an UOV system in degree $2\leq D\leq n$ over $\FF_2$. 
For a UOV system one always has $m<n$, since $n-m$ is the number of vinegar variables. If $D=n$, then $m=\varsigma_1=1$. Else, $\varsigma_1={n-1\choose D-1}\geq n-1$. In both cases, one has $m\in(0,\varsigma_1]$ and $k=1$. Therefore, by Corollary~\ref{cor:D>q} $$d_{\reg}(\mathcal{F}\cup\{x_1^2-x_1,\ldots,x_n^2-x_n\})\leq n$$ and if either $D=n$ and $m=1$, or $D\geq 3$ and $m=n-1$, then $$d_{\reg}(\mathcal{F}\cup\{x_1^2-x_1,\ldots,x_n^2-x_n\})\leq n-1.$$
\end{example}

One can also approach a UOV system by assigning random values to $n-m$ variables and trying to solve for the rest. This yields a system of $m$ equations in $m$ variables, to which one can still add the field equations. We compute the bound for a quadratic UOV system.

\begin{example}
Let $\mathcal{F}$ be an UOV system in degree $D=2$ over $\FF_2$. 
Assign random values to $n-m$ variables to obtain a system of $m$ quadratic equations in $m$ variables. If $m>2$, then $m\leq\dim_{\FF_2}\FF_2[x_1,\ldots,x_m]_2-\dim_{\FF_2}\langle x_1^2,\ldots,x_m^2\rangle$, so one expects that the specializations of $f_1^{\topp},\ldots,f_m^{\topp}$ remain linearly independent modulo $\langle x_1^2,\ldots,x_m^2\rangle$. 
Since $\varsigma_1=m-1$ and $\varsigma_2=2m-3$, then $m\in(\varsigma_1,\varsigma_2]$ and $k=2$. Therefore, by Corollary~\ref{cor:D>q} the degree of regularity is bounded from above by $m-1$.
\end{example}

In Theorem \ref{teor:DD:LPP} and Theorem \ref{teor:D>q} we consider a system $\mathcal{F}$ of equations of the same degree~$D$ and provided upper bounds on the degree of regularity of $\mathcal{F}$ under the assumption that $\mathcal{F}$ is regular in degree $D$, and of $\mathcal{F}\cup\{x_1^q-x_1,\ldots,x_n^q-x_n\}$ if $D\geq q$. 
Since our results also apply to semiregular systems, the bounds that we obtain are necessarily larger than the value of the degree of regularity of a semiregular system of equations of the same degree. A common approach in the cryptoanalysis of a multivariate scheme is to assume that the corresponding system is semiregular. Under such an assumption, one can compute the value of the degree of regularity of the system. We stress that, even though the bounds that we obtain are less tight than those that one obtains by assuming that a system is semiregular, they have the advantage of being provable, since they just require that the degree of regularity is finite and do not rely on the heuristic assumption that the system is semiregular.

\end{document}